\documentclass[lettersize,journal]{IEEEtran}

\usepackage{epsfig} 
\usepackage{graphicx}
\usepackage{amsmath} 
\usepackage{amssymb}  

\usepackage[utf8]{inputenc}
\usepackage{upgreek}
\usepackage{amsfonts}
\usepackage{mathtools}
\usepackage{xspace}
\usepackage[mathscr]{euscript} 
\usepackage{dirtytalk}
\usepackage{tikz}
\usetikzlibrary{positioning,arrows.meta,quotes,chains,shapes,fit,calc}

\let\labelindent\relax
\usepackage[inline]{enumitem}

\usepackage[caption=false,font=footnotesize]{subfig}

\usepackage{tabularx}
\usepackage{txfonts} 

\usepackage{cases} 
\usepackage[noadjust]{cite} 

\usepackage[ruled,norelsize]{algorithm2e} 

\makeatletter
\newenvironment{HP}[1][H]{%
   \begin{algorithm}[#1]%
  }{\end{algorithm}}
\makeatother  

\makeatletter
\newcommand{\removelatexerror}{\let\@latex@error\@gobble}
\makeatother

\usepackage{setspace} 
\usepackage{siunitx}
\usepackage{cases}
\usepackage{empheq}
\usepackage{microtype} 

\newcommand{\dl}{\text{\upshape\textsf{d{\kern-0.03em}L}}\xspace}
\newcommand{\keyx}{KeYmaera\,X\xspace}

\newcommand{\dlold}{\text{\upshape\textsf{d{\kern-0.05em}$\mathcal{L}$}}\xspace}

\newcommand{\stateset}{\ensuremath{\mathscr{S}}\xspace}

\newcommand{\mdl}{\ensuremath{\mathscr{M}}\xspace}

\newcommand{\real}{\ensuremath{\mathbb{R}}\xspace}
\newcommand{\rational}{\ensuremath{\mathbb{Q}}\xspace}

\newcommand{\Tau}{\mathsf{T}} 

\newcommand{\aMinN}{\ensuremath{a_n^{min}}\xspace}
\newcommand{\aMaxN}{\ensuremath{a_n^{max}}\xspace}
\newcommand{\aMinS}{\ensuremath{a_s^{min}}\xspace}
\newcommand{\aReq}{\ensuremath{\operatorname{a^{req}}}\xspace}
\newcommand{\ath}{\ensuremath{\operatorname{a^{th}}}\xspace}
\newcommand{\aReqi}{\ensuremath{\operatorname{a_0^{req}}}\xspace}
\newcommand{\aReqT}{\ensuremath{\operatorname{a_\Tau^{req}}}\xspace}

\newcommand{\opif}{\ensuremath{\operatorname{if}\xspace}}

\newcommand{\opelse}{\ensuremath{\operatorname{else}\xspace}}

\newcommand{\opmsd}{\operatorname{msd}}

\newcommand{\msd}{\ensuremath{\opmsd\left(t \right)}\xspace}

\newcommand{\msdt}{\ensuremath{\opmsd\left(\Tau \right)}\xspace}

\newcommand{\msdi}{\ensuremath{\opmsd\left(0 \right)}\xspace}

\newcommand{\safe}{\ensuremath{\mathit{ok}}\xspace} 
\newcommand{\assume}{\ensuremath{\mathit{init}}\xspace}
\newcommand{\guarantee}{\ensuremath{\mathit{guarantee}}\xspace}

\newcommand{\ego}{ego-vehicle\xspace}

\newcommand{\env}{\emph{env}\xspace}
\newcommand{\ctrl}{\emph{ctrl}\xspace}
\newcommand{\plant}{\emph{plant}\xspace}

\DeclarePairedDelimiter\abs{\lvert}{\rvert}%

\newcommand{\pos}{\ensuremath{\mathit{x}}\xspace}
\newcommand{\vel}{\ensuremath{\mathit{v}}\xspace}
\newcommand{\delt}{\ensuremath{\mathit{t}}\xspace}
\newcommand{\acc}{\ensuremath{\mathit{a}}\xspace}

\newcommand{\mone}{\ensuremath{\opmsd_1}\xspace}
\newcommand{\mthree}{\ensuremath{\opmsd_3}\xspace}
\newcommand{\mfive}{\ensuremath{\opmsd_5}\xspace}

\usepackage{amsthm}
\newtheorem{theorem}{Theorem}

\theoremstyle{remark}

\newtheorem{remark}{Remark}
\newtheorem{sg}{SG}
\newtheorem{fsr}{FSR}

\newtheorem{newfsrin}{FSR}
\newenvironment{newfsr}[1]
  {\renewcommand\thenewfsrin{#1}\newfsrin}
  {\endnewfsrin}

\usepackage{lipsum}
\usepackage{hyperref}
\begin{document}

\title{Formal Development of Safe Automated Driving using Differential Dynamic Logic}

\author{Yuvaraj Selvaraj, Wolfgang Ahrendt, and Martin Fabian
\thanks{This work was supported by FFI, VINNOVA under grant number 2017-05519, \textit{Automatically Assessing Correctness of Autonomous Vehicles--Auto-CAV} and partly supported by the Wallenberg AI, Autonomous Systems and Software program (WASP) funded by the Knut and Alice Wallenberg Foundation.}
\thanks{Yuvaraj Selvaraj is with Zenseact, 417 56 Gothenburg, Sweden,
and also with the Department of Electrical Engineering, Chalmers
University of Technology, 412 96 Gothenburg, Sweden (e-mail:
yuvaraj.selvaraj@zenseact.com).}%
\thanks{Wolfgang Ahrendt is with the Department
of Computer Science and Engineering, Chalmers University of
Technology, 412 96 Gothenburg, Sweden (e-mail: ahrendt@chalmers.se).}
\thanks{Martin Fabian is with the Department of Electrical Engineering, Chalmers University of Technology, 412 96 Gothenburg, Sweden (e-mail: fabian@chalmers.se).}%
}




\maketitle

\begin{abstract}
The challenges in providing convincing arguments for safe and correct behavior of automated driving (AD) systems have so far hindered their widespread commercial deployment. Conventional development approaches such as testing and simulation are limited by non-exhaustive analysis, and can thus not guarantee correctness in all possible scenarios. Formal methods is an approach to provide mathematical proofs of correctness, using a model of a system, that could be used to give the necessary arguments. This paper investigates the use of differential dynamic logic and the deductive verification tool \keyx in the development of an AD feature. Specifically, formal models and safety proofs of different design variants of a Decision \& Control module for an in-lane AD feature are presented. In doing so, the assumptions and invariant conditions necessary to guarantee safety are identified, and the paper shows how such an analysis helps during the development process in requirement refinement and formulation of the operational design domain. Furthermore, it is shown how the performance of the different models is formally analyzed exhaustively, in \emph{all} their allowed behaviors.
\end{abstract}

\begin{IEEEkeywords}
Automated driving, formal methods, safety argument, deductive verification, logic. 
\end{IEEEkeywords}

\section{INTRODUCTION} \label{sec:intro}
\emph{Automated driving} (AD) has many potential benefits~\cite{litman2017autonomous}, such as reducing road traffic accidents, reducing driver stress, improving energy efficiency, availing independent mobility to people who cannot or should not drive, etc. The level of automation can range from supervised support features like \emph{advanced driver assistance systems} (ADAS) to unsupervised AD features~\cite{saej3016}. However, there are many barriers for the commercial deployment of full autonomy in road vehicles; particularly crucial and technically challenging is providing convincing arguments for safe and correct behavior of AD systems~\cite{koopman2016challenges,koopman2017autonomous}.

Several approaches could be used to argue about AD safety~\cite{koopman2019credible}. Though not specifically intended for AD, a well-established approach in the automotive industry is to show conformance to safety standards such as ISO 26262~\cite{iso201826262}, which addresses hazards due to malfunctioning behavior, and ISO/PAS 21448~\cite{isopas2019}, which addresses hazards due to unintended behavior. The high complexity of AD systems makes it difficult to satisfy some safety objectives in the development activities recommended by such standards. For instance, providing sufficient evidence for the correctness of the safety requirements and their verification in each phase of the development process is a significant challenge~\cite{bergenhem2015reach,koopman2016challenges}.

One way to tackle this challenge is to restrict the op\-er\-a\-tional environment of the AD system through an \emph{operational design domain} (ODD), defined in SAE J3016~\cite{saej3016} as \say{Operating conditions under which a given driving automation system or feature thereof is specifically designed to function, including [\dots\!] geographical, and time-of-day restrictions [\dots\!].}
Thus, the ODD limits the scope of development activities like \emph{hazard analysis and risk assessment} (HARA), requirement refinement, verification, etc. Even so, it is challenging to sufficiently identify the ODD such that it can provide safety requirements to implement and verify~\cite{gyllenhammar2020towards,koopman2019many}.

\subsection{Illustrative Example}\label{sec:example}
To emphasize the challenges involved, consider an autonomous vehicle (hereafter referred to as the \emph{\ego}) that offers in-lane unsupervised AD at speeds up to $\SI{100}{\km/\hour}$. This is realized using the simplified functional architecture in Fig.~\ref{fig:archgen}, where \begin{enumerate*}[label=(\roman*)]
  \item \emph{Sense}, perceives the environment and provides information such as the vehicle state, traffic state, etc.;
  \item \emph{Decision \& Control},  decides on when and how to act (e.g. accelerating/steering commands); and
  \item \emph{Actuate}, executes the decisions using the respective actuators.
\end{enumerate*}

\tikzset{
     block/.style={rectangle, draw, text width=5.5em,
                   text centered, rounded corners, minimum height=3em},
     arrow/.style={-{Stealth[]},thick}
     }

\begin{figure}[!t] 
\centering
\begin{tikzpicture}
    \node [block] (sense) {Sense};
    \node [block,right=of sense] (dc)  {Decision \& Control};
    \node [block,right=of dc] (act)  {Actuate};
    
    \draw [arrow] (sense) -- (dc);
    \draw[arrow] (dc) -- (act);
\end{tikzpicture}
\caption{A simplified functional architecture for an AD system}
\label{fig:archgen}
\end{figure}
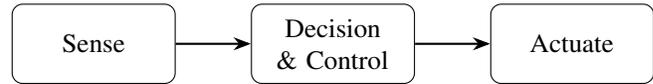

The feature definition above is overly general and could be restricted by adding ODD conditions as \say{[\dots] speeds up to $\SI{100}{\km/\hour}$ in Gothenburg during daytime.} Ideally, all such conditions should be explicitly defined at the start of the development process. However, it is notably difficult to do so, due to the wide range of parameters that affect the ODD and the lack of a standardized procedure to formulate the ODD. Still, these conditions are crucial, as they form the basis for the subsequent HARA to obtain the safety requirements, which must be unambiguously specified. This highlights the need for a systematic method to identify such conditions and refine them throughout the design process as envisioned in~\cite{colwell}. Without presenting a detailed HARA in this paper, a possible high-level safety requirement, termed a \emph{safety goal} (\emph{SG}) within ISO~26262, is:   
\begin{sg}\label{sg:1}
\emph{The \ego shall maintain a safe distance to other objects in front such that a collision is at all times avoided.}
\end{sg}

As \emph{SG}\,\ref{sg:1} is to be realized using Fig.~\ref{fig:archgen}, a refinement of it, using the ISO~26262 requirement refinement process, could be the \emph{functional safety requirement} (\emph{FSR}):
\begin{fsr}\label{fsr:dc}
\emph{Decision \& Control shall at all times output a safe acceleration request to avoid a collision with any object in front.}
\end{fsr}

The successful verification of \emph{SG}\,\ref{sg:1} and \emph{FSR}\,\ref{fsr:dc} is necessary to make an overall safety argument. But, their generality presents difficulties in the verification and consequently for the safety argument. For instance, \emph{SG}\,\ref{sg:1} and \emph{FSR}\,\ref{fsr:dc} require the ego-vehicle to avoid collisions with other objects \emph{at all times}. Obviously, this is desirable, but difficult to verify for all the situations that the \ego might encounter. More pragmatic is to guarantee that the safety requirements are fulfilled under specific ODD assumptions such as system dynamics, object behavior, system limitations, etc.  

Concretely, consider a case where the controller is to guarantee that the \ego stops at or before a critical position. This can model stopping at a traffic light, or not colliding with a leading vehicle, \emph{SG}\,\ref{sg:1}. Fig.~\ref{fig:ill1} illustrates the \ego's position and velocity on the \emph{position}-\emph{velocity} plane for a given set of inputs and parameter values. The red dashed lines represent the  critical position and zero velocity. A subset of the states, the red region, are the \emph{forbidden} states where the \ego exceeds zero velocity beyond the critical position. Though the other states do not directly violate the requirement, some of them must be avoided as due to the system dynamics the controller cannot guarantee to avoid the forbidden states. Thus, the state-space is partitioned into \emph{admissible} and \emph{inadmissible} states. The admissible states, the purple region, are the states from where the control system can guarantee that the forbidden states are avoided. The inadmissible states, the white region, are the states from where reaching the forbidden region cannot be avoided. Identifying the ODD assumptions and invariant conditions that characterize such admissible and inadmissible states is crucial to the safety verification of \emph{SG}\,\ref{sg:1} and \emph{FSR}\,\ref{fsr:dc}. Furthermore, they are also necessary for obtaining specifications at the hardware and software level to implement the controller. In this regard, this paper focuses on how formal methods can aid in the development of a safe AD system.

\begin{figure}[!t] 
\centering
\includegraphics[width=0.9\linewidth, keepaspectratio]{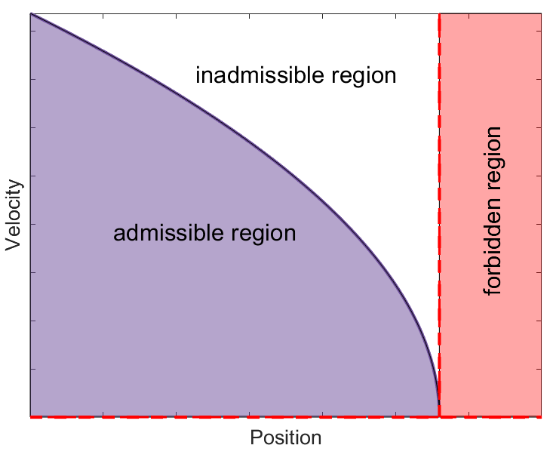}
\caption{An illustration of a two-dimensional state-space representation of a control system realizing the example of Section~\ref{sec:example}.}
\label{fig:ill1}
\end{figure}

\subsection{Contributions}
This paper investigates the use of \emph{differential dynamic logic} (\dl)~\cite{platzer2008differential,platzer2012logics,platzer2018logical} and the associated deductive verification tool \keyx~\cite{fulton2015keymaera}, to formally model and verify correctness of different design variants of a Decision \& Control module for an in-lane AD feature. Deductive verification and the expressiveness of \dl enable symbolic verification of parametric systems so that a proof of a single model can be used in the safety argument of a large number of concrete implementations. 

The logic \dl supports the specification and verification of \emph{hybrid systems}, which are mathematical models of systems that combine discrete and continuous dynamics. Formal analysis of any AD feature typically requires reasoning about continuous state variables (e.g. position, velocity) and discrete decisions (e.g. to brake, to accelerate), thus making \dl particularly suitable to achieve a good trade-off between formal models close to reality and a tractable analysis. Additionally, as shown in this work, the formulas of \dl used to specify safety requirements like \emph{FSR}\,\ref{fsr:dc} can be represented in the form of \emph{assume-guarantee} requirements to permit compositional reasoning based on principles of contract-based design~\cite{sangiovanni2012taming}, thereby reducing design and verification complexity. 

To summarize, this paper:
\begin{enumerate}[leftmargin=0.44cm]
    \item presents formal models of different design variants of a Decision \& Control module for an in-lane AD feature and proves their correctness with respect to safety requirements. The models differ in the decision-making logic used and it is shown how formal verification identifies safety-critical edge cases that arise due to subtle changes in the decision-making logic (Models~\ref{alg:sm1}-\ref{alg:sm5} in Section~\ref{sec:models});

    \item provides safety proofs in the form of assume-guarantee requirements in \dl as evidence for the safety argument of the AD feature (Theorems~\ref{th:1}-\ref{th:5} in Section~\ref{sec:models}). In doing so, during the formal analysis, the necessary assumptions and the invariant conditions required to guarantee safety are formulated. Furthermore, it is discussed how such an analysis could help in requirement refinement during the practical industrial development process (Section~\ref{sec:identification});

    \item shows how formal proofs can be used to analyze the performance (with respect to a defined metric) of the different models, in \emph{all} behaviors allowed by the respective model, in contrast to conventional approaches that involve independent analysis of different (non-exhaustive) sets of scenarios (Section~\ref{sec:performance}). 
\end{enumerate}


\section{PROBLEM FORMULATION} \label{sec:prob}
The focus of this paper is the safety verification of the Decision \& Control module for an in-lane AD feature. While the simplified functional architecture in Fig.~\ref{fig:archgen} is representative of any AD feature, the Decision \& Control needs to be refined to concretely formulate the verification problem.  

A typical driving task can be divided into three levels of decision-making~\cite{michon1985critical}: \emph{strategic} (e.g. route planning over long time horizon), \emph{tactical} (e.g. maneuvering over a few seconds), and \emph{operational} (e.g. speed control on milliseconds level). For an AD feature, SAE J3016~\cite{saej3016} requires the AD system to completely perform the \emph{dynamic driving task} (DDT), defined as \say{all of the real-time operational and tactical functions required to operate a vehicle in on-road traffic, excluding the strategic functions [\dots]}. Accordingly, this work considers only the tactical and operational levels of the Decision \& Control module. 

While performing the DDT, the \ego is required to handle a variety of in-lane scenarios like maintaining a safe distance to a lead vehicle, stopping for an obstacle, etc., and is subjected to different constraints to ensure safety, comfort, etc. This driving task could be solved by any algorithm (e.g. feedback control law, reinforcement learning), some of which might be hard to analyze and verify. One way to make the safety verification tractable is to separate the nominal functions from the safety functions by means of the \emph{Nominal Controller} and the \emph{Safety Controller} as shown in Fig.~\ref{fig:funcarch}. The Nominal Controller represents any algorithm solving the nominal driving task and requests a nominal acceleration $a_n$. The Safety Controller ensures that only safe decisions are communicated to the \emph{Vehicle Control} by evaluating $a_n$ and calculating a safe acceleration $a_s$, thereby satisfying the safety requirement. Thus, the safety verification can be reformulated as a verification problem of a simpler component, the Safety Controller.   

As shown in Fig.~\ref{fig:funcarch}, in addition to $a_n$, the Safety Controller receives information about the vehicle state from the \emph{Sense} module and information about the safety constraint given by a critical position $x_c$ and a critical velocity $v_c$ (see Section~\ref{sec:models}) from the \emph{Situation Assessment} function block. With this information, the Safety Controller calculates $a_s$, which is considered safe if it fulfills the safety constraint. The {Vehicle Control} then, takes $a_s$ as a reference and calculates trajectories for vehicle motion to solve the operational function of the DDT. Finally, these motion trajectories are executed by the \emph{Actuate} module through the various actuators in the \ego.  

The architecture in Fig.~\ref{fig:funcarch} permits to abstract away from the possibly complex design of the Nominal Controller, and verify that $a_s$ always respects the safety constraint by reasoning about the decision-making in the Safety Controller. Moreover, specifying a safety requirement like \emph{FSR}\,\ref{fsr:dc} using a safety constraint from a separate functional block provides the flexibility to dynamically calculate constraints for a variety of situations. Furthermore, this architecture also enables modular reasoning and therefore the verification approach discussed in this paper can easily be adapted to other AD features outside the scope of in-lane AD. Though this approach makes the safety verification efficient and tractable, one might rightly argue that certain assumptions have to be made about the Nominal Controller and possibly other functions within Decision \& Control to guarantee that the safety requirements are satisfied. The rest of this paper deals with this safety verification problem where, the Decision \& Control module of Fig.~\ref{fig:funcarch} is formally verified, and in doing so, the assumptions and invariant conditions necessary to guarantee safety are identified.

\tikzset{
     block/.style={rectangle, draw,
                   text width=5em,
                   text centered, rounded corners, minimum height=3em, on grid},
     arrow/.style={-{Stealth[]},thick},
     no arrow/.style={-,every loop/.append style={-},thick},
     smallblock/.style={draw, rectangle,text width=3em, text centered, rounded corners, minimum height=2em}
     }

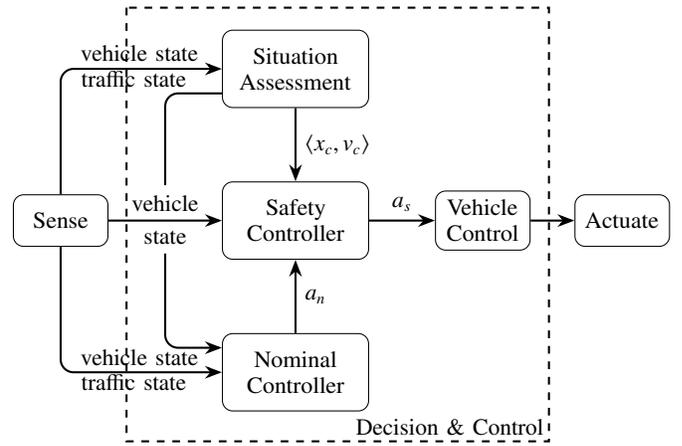
\begin{figure}[!t]
    \centering
    \resizebox{\linewidth}{!}{%
    \begin{tikzpicture}[auto, font=\small]
    
    \node [smallblock] (sense) at (-3.5,0) {Sense};
    \node [block, right=of sense] (safety) at (-1.3,0) {Safety Controller}; 
    \node [smallblock, right=of safety](vehicle) at (0.6,0) {Vehicle Control};
    \node [smallblock, right =of vehicle](actuate) at (2.5,0){Actuate};
    \node [block, above = of safety.north](sa) {Situation Assessment};
    \node [block, below=of safety.south](nominal) {Nominal Controller};
    \draw[thick,dashed] ($(sa.north west)+(-1.3,0.3)$)  rectangle ($(vehicle.south east)+(0.25,-2.6)$);
    
    \draw [arrow] (sense) --  node[name=v,above] {vehicle}
    node[name=s,below] {state}
    (safety);
    \draw [arrow] (safety) -- node[name=as] {$a_s$} (vehicle);
    \draw [arrow] (vehicle) -- (actuate);
    \draw [arrow] (nominal) -- node[name=an,right,midway] {$a_n$}(safety); 
    \draw [arrow] (sa) -- node[name=con]{$\langle x_c,v_c\rangle$}(safety);
    \draw [arrow,rounded corners=5pt] (sense) |- (sa);
    \draw [arrow,rounded corners=5pt] (sense) |-  (nominal); 
    \draw [no arrow, rounded corners=5pt]($(sa.south west)+(0,0.2)$)-| (v);
    \draw [arrow,rounded corners=5pt](s) |- ($(nominal.north west)+(0,-0.2)$);

    \node[name=vs1] at ($(sa.north west)+(-1.125,-0.325)$) 
    {vehicle state};
    \node[name=ts1] at ($(sa.north west)+(-1.2,-0.675)$) 
    {traffic state};
    \node[name=vs2] at ($(nominal.north west)+(-1.125,-0.325)$) 
    {vehicle state};
    \node[name=ts2] at ($(nominal.north west)+(-1.2,-0.675)$) 
    {traffic state};
    \node[name=ts2] at
    ($(vehicle.south east)+(-1.1,-2.4)$)
    {Decision \& Control};
    
    \end{tikzpicture}
    }
    \caption{A refinement of the functional architecture in Fig.~\ref{fig:archgen}.}
    \label{fig:funcarch}
\end{figure}

\section{PRELIMINARIES} \label{sec:prelims}
The logic \dl supports the specification and verification of \emph{hybrid systems}, that is, mathematical models of systems that combine discrete dynamics (behavior that changes discretely) with continuous dynamics (behavior that changes continuously with time). This makes \dl particularly suitable for the modeling and verification of AD systems since it can reason about continuous state variables (e.g. position, velocity) and discrete decisions (e.g. to brake, to accelerate). 

To model hybrid systems, \dl has the notion of \emph{hybrid programs} (HP) that consist of different program statements, including differential equations to describe continuous behavior. HPs are defined by the following grammar, where $\upalpha,\,\upbeta$ are HPs, $x$ is a variable, $e$ is a term\footnote{terms are polynomials with rational coefficients defined by $e,\,\Tilde{\!e}\;\Coloneqq x\;|\;c\in\rational\;|\;e+\,\Tilde{\!e}\;|\;e\cdot\,\Tilde{\!e}$} possibly containing $x$, and $P,Q$ are formulas of first-order logic of real arithmetic\footnote{First-order logic formulas of real arithmetic are defined by $P,\,Q\Coloneqq e\ge\,\Tilde{\!e}\;|\;e=\,\Tilde{\!e}\;|\;\neg P\;|\;P \land Q\;|\;P \lor Q\;|\;P\to Q\;|\;P\leftrightarrow Q\;|\;\forall x P\;|\;\exists x P$ }:
\begin{equation*}\label{eqn:HPgrammar}
\upalpha,\,\upbeta\,\Coloneqq x \coloneqq e \, | \, x\coloneqq * \, | \, ?P \, | \, x'=f(x)\,\&\,Q \, | \, \upalpha\,\cup\,\upbeta \, | \, \upalpha;\,\upbeta \, | \, \upalpha^* 
\end{equation*}
A summary of the program statements in HP and their informal semantics~\cite{platzer2018logical}, is given in Table~\ref{tab:hp}.

\begin{table}[!ht]\small
    \caption{Statements of HPs~\cite{platzer2018logical}. $P,Q$ are first-order formulas, $\upalpha,\upbeta$ are HPs.}
    \label{tab:hp}
    \centering
    \begin{tabularx}{\linewidth}{p{6em} X}
         \hline 
         Statement & \multicolumn{1}{c}{Effect} \\
         \hline
         $x\coloneqq e$ & discrete assignment of the current value of term $e$ to variable $x$\\
         $x\coloneqq *$ & nondeterministic assignment of an arbitrary real number to $x$\\
         $?P$ & continue if first-order formula $P$ holds in the current state, abort otherwise\\
         $x'=f(x)\,\&\,Q$ & follow differential equation $x'=f(x)$ within evolution domain $Q$ for any duration\\
         $\upalpha \cup \upbeta$  & nondeterministic choice, follow either $\upalpha$ or $\upbeta$ \\
         $\upalpha;\,\upbeta$ & sequential composition where $\upbeta$ starts after $\upalpha$ finishes\\
         $\upalpha^*$ & nondeterministic repetition, repeat $\upalpha$ $n$ times for any $n \in \mathbb{N}_0$ \\
         \hline
    \end{tabularx}
\end{table}

Each HP $\upalpha$ is semantically interpreted as a reachability relation $\uprho(\upalpha) \subseteq \stateset \times \stateset$, where $\stateset$ is the set of states and a state $s \in \stateset$ is defined as a mapping from the set of variables to $\real$. The \emph{test} action $?P$ has no effect in a state where $P$ is true. If however $P$ is false when $?P$ is executed, then the current execution of the HP \emph{aborts}, meaning that the entire current execution is removed from the set of possible behaviors of the HP. 
Test actions are often used together with non-deterministic assignment, like $a_n\coloneqq *;\; ?\left(-\aMinN \le a_n \le \aMaxN\right)$. This expresses that $a_n$ is assigned an arbitrary value, which is then tested to be within the bounds $-\aMinN$ and $\aMaxN$. If the value of $a_n$ is outside those bounds, that branch of execution is aborted.
This can model that some external component/environment chooses $a_n$ to be within the given bounds. Furthermore, test actions can be combined with sequential composition and the choice operator to define \emph{if-statements} from classical programming languages as: 
\begin{equation}\label{eq:ifelse}
    \opif\,(P)\; \upalpha \opelse \upbeta \equiv (?P;\,\upalpha) \cup (?\neg P;\,\upbeta)
\end{equation}

HPs model continuous dynamics as $ x'=f(x)\,\&\,Q$, which describes the \emph{continuous evolution} of $x$ along the differential equation system $x'=f(x)$ for an arbitrary duration (including zero) within the \emph{evolution domain constraint} $Q$. The evolution domain constraint applies bounds on the continuous dynamics and are first-order formulas that restrict the continuous evolution within that bound. $x'$ denotes the time derivative of $x$, where $x$ is a vector of variables and $f(x)$ is a vector of terms of the same dimension.

The nondeterministic actions (assignment $x\coloneqq *$, choice $\upalpha\,\cup\,\upbeta$, and repetition $\upalpha^*$) help address two critical aspects in the safety verification:
\begin{enumerate*}[label=(\roman*)]
    \item they can describe unknown behavior, which is typically the case in modeling the highly unpredictable environment for AD systems;
    \item they can abstract away implementation specific details and thus reduce the dependency of the proof on such details. 
\end{enumerate*} For example, to reason about the correctness of the Safety Controller, the nominal acceleration $a_n$ can be modeled with a nondeterministic assignment together with a test action as described above. Such a model is realistic for any implementation of the Nominal Controller and therefore makes the formal analysis independent of changes in the implementation.

The formulas of \dl include formulas of first-order logic of real arithmetic and the modal operators $[\upalpha]$ and $\langle\upalpha\rangle$ for any HP $\upalpha$~\cite{platzer2012logics,platzer2018logical}. The formulas of \dl are defined by the following grammar ($\upphi,\,\uppsi$ are \dl formulas, $e,\,\Tilde{\!e}$ are terms, $x$ is a variable, $\upalpha$ is a HP):
\begin{equation}\label{eq:dlgrammar}
    \upphi,\,\uppsi \Coloneqq e=\,\Tilde{\!e}\;|\;e\ge \,\Tilde{\!e}\;|\;\neg \upphi\;|\;\upphi \land \uppsi\;|\;\forall x\,\upphi\;|\;[\upalpha]\,\upphi
\end{equation}

The \dl formula $[\upalpha]\,\upphi$ expresses that all non-aborting runs of HP $\upalpha$ (i.e., the runs where all test actions are successful) lead to a state in which the the \dl formula $\upphi$ is true. The \dl formula $\langle\upalpha\rangle\,\upphi$ says that there exists some non-aborting run leading to a state where $\upphi$ is true.  $\langle\upalpha\rangle\,\upphi$ is the dual to $[\upalpha]\,\upphi$, defined as $\langle\upalpha\rangle\,\upphi \equiv \neg[\upalpha]\neg\upphi$. Similarly, operators $>,\le,<,\lor,\to,\leftrightarrow,\exists x$ are defined using combinations of the operators in (\ref{eq:dlgrammar}). 

To specify the correctness of the HP $\upalpha$, we use a \dl formula of the form $\upphi\to[\upalpha]\,\uppsi$, which expresses that if the formula $\upphi$ is true in the initial state, then all (non-aborting) runs of $\upalpha$ only lead to states where formula $\uppsi$ is true. In our context of AD, this can easily be translated to a \dl formula:
\begin{equation} \label{eq:asgu}
(\assume) \to [\,\mdl\,]\,(\mathit{guarantee})
\end{equation}
where \mdl is the HP describing the Decision \& Control module, \emph{init} is the initial condition, and \emph{guarantee} is the safety requirement to be verified. The following sections describe how this \dl formula is further refined and formally proven using \keyx, which implements a verification technique for \dl~\cite{platzer2012logics,platzer2018logical,fulton2015keymaera}.

\section{MODELS AND PROOFS} \label{sec:models}

As discussed in Section~\ref{sec:prob}, we formally analyze the safety of the Decision \& Control module by reasoning about the decision-making in the Safety Controller, and in doing so, formulate the necessary assumptions and invariant conditions to guarantee safety. The objective of the Safety Controller is to calculate a safe acceleration value $a_s$ that always fulfills the safety constraint given by the pair $\langle x_c,v_c\rangle$, the critical position and critical velocity, respectively. The requirement to guarantee safety of the \ego is to not have a velocity higher than $v_c$ at or beyond $x_c$. 

The \ego's (longitudinal)\footnote{Only in-lane scenarios are dealt with, so the terms position, velocity and acceleration describe longitudinal vehicle motion, unless otherwise noted.} motion is described by the continuous time kinematic equations:
\begin{equation}\label{eq:vehicleode}
    \frac{dx}{dt} = v,\; \frac{dv}{dt} = a,
\end{equation}
where position $x$ and velocity $v$ are the state variables and the acceleration $a$ is piece-wise constant. Fig.~\ref{fig:sim2} illustrates an example simulation of the ego-vehicle's motion model. 

Furthermore, we consider four system parameters in the models: the maximum and minimum acceleration limits of the 
Nominal Controller given by \aMaxN and \aMinN respectively, the maximum braking limit of the Safety Controller \aMinS, and the sampling time $\Tau$. Since the Nominal Controller is subject to different constraints (e.g. comfort constraints) during nominal driving conditions, its braking capability is less than the vehicle's maximum braking capability. In contrast, the Safety Controller can use the vehicle's maximum braking capability and can therefore brake harder than the Nominal Controller, i.e., $\abs*{\,\aMinS\,} > \abs*{\,\aMinN\,}$. 

In the models~\ref{alg:sm1}-\ref{alg:sm5}, the respective \plant~(\eqref{eq:m1:3},~\eqref{eq:m2:3},~\eqref{eq:m3:3},~\eqref{eq:m4:3}, \eqref{eq:m5:3}) models the continuous dynamics together with the evolution domain constraint. The \ego's motion described in~\eqref{eq:vehicleode} is modeled as $x'=v,\ v'=a_s$ where $a_s$ is the safe acceleration value from the Safety Controller. The evolution domain constraint $v\ge 0$ restricts the continuous evolution to only non-negative velocities. In addition, the \plant models a clock variable $\tau$ that evolves along the differential equation $\tau'=1$, is bound by the domain constraint $\tau \le \Tau$, and is set to $\tau=0$ before every evolution of the \ego's motion. Intuitively, $\tau$ represents the controller execution/sampling time. The constraint $\tau\le \Tau$ accounts for non-periodic sampling as it allows evolution for any amount of time not longer than $\Tau$. In every execution loop of~\eqref{eq:m1:1}, first \env and \ctrl get executed instantaneously, the clock is reset, and then the \plant evolves for at most $\Tau$ time before the loop either repeats again or terminates. 

The rest of this section describes the models and proofs for different design variants of the Decision \& Control module. The models differ in the decision logic used in the Safety Controller. We first prove safety of Model~\ref{alg:sm1}, with a conservative safety controller and the critical velocity $v_c=0$. This model is then generalized to $v_c \ge 0$ (Model~\ref{alg:sm2}). Then, we extend the proofs to models with different threat metrics (models~\ref{alg:sm3}-\ref{alg:sm5}). Furthermore, we also remark on the implications of some modeling choices on the safety argument. 

\subsection{Model 1: Conservative with \texorpdfstring{$v_c = 0$}{vCrit zero}}\label{sec:sm:sr}

In Model~\ref{alg:sm1}, the \dl formula  \eqref{eq:m1:1} refines $\mdl$ of~\eqref{eq:asgu} as a nondeterministic repetition of three sequentially composed HPs to represent a typical controller-plant model: \env~\eqref{eq:m1:7}--\eqref{eq:m1:7a}, \ctrl~\eqref{eq:m1:4}--\eqref{eq:m1:6}, and \plant~\eqref{eq:m1:3}. The \env and \ctrl HPs comprise discrete time sensor inputs and control outputs and are hence modeled as discrete assignments in the HP, while the \plant models the continuous time behavior of the \ego. The nondeterministic repetition modeled by the $^*$ means that the sequential composition (\env; \ctrl; \plant) is repeatedly executed an arbitrary number of times, possibly zero. The \guarantee~\eqref{eq:m1:2} describes the requirement to be verified, i.e., the \ego's velocity is equal to the critical velocity (zero in this particular model) at or beyond the critical position. The \dl formula~\eqref{eq:m1:1} asserts that, here, the forbidden region according to Fig.~\ref{fig:ill1} consists of all violations of \guarantee. 

As described in Section~\ref{sec:prob}, the Safety Controller calculates $a_s$ based on the nominal acceleration $a_n$, the safety constraint $\langle x_c, v_c \rangle$, and the current \ego state given by the position $x$ and velocity $v$. During the driving task, it is logical for the Safety Controller to concur with the Nominal Controller as long as $a_n$ does not compromise safety. Otherwise, the Safety Controller ensures safety through maximal braking. In~\eqref{eq:m1:4}, \ctrl models this decision-making in the Safety Controller by checking if the \safe condition~\eqref{eq:m1:5} is true. If so, $a_n$ does not compromise safety and is assigned to $a_s$. Else, the maximum braking $-\aMinS$ is assigned to $a_s$.

To assess whether $a_n$ is safe or not, a suitable threat metric is needed. In this model, a threat metric in the distance domain, namely the \emph{minimal safe distance}, \msd, is defined in~\eqref{eq:m1:6}. Here, $\msd$ is calculated based on current velocity $v$, time interval $t$, \aMaxN, and \aMinS, using a worst-case assumption for the \ego behavior. During any given time interval $t$, the worst possible behavior of the \ego, w.r.t. the safety constraint, is to accelerate with maximum value \aMaxN. However, such a behavior should be admissible only if the Safety Controller can, after the interval $t$, fulfill the safety constraint by maximal braking \aMinS. This is captured by~\eqref{eq:m1:6}, where \msd is given by the sum of the distance traveled from the current state by accelerating with $\aMaxN$ for time interval $t$ and from there on, the distance traveled by braking with $\aMinS$ until $v=0$ is reached, i.e., the \ego is completely stopped.   

\begin{figure}[!t]
 \removelatexerror
  \begin{HP}[H]
    \DontPrintSemicolon
    \setstretch{0.5}
    \caption{Conservative with $v_c = 0$}
    \label{alg:sm1}
    \SetAlgoLined
    
    \begin{flalign} \label{eq:m1:1}
        (\assume) \to [(\env;\,\ctrl;\,\plant)^*]\,(\guarantee) &&
    \end{flalign}
    \begin{flalign}\label{eq:m1:2}
        \guarantee \triangleq (x \ge x_c \to v = 0) &&
    \end{flalign}
    \begin{flalign}\label{eq:m1:3}
        \plant \triangleq \tau\coloneqq 0;\, x' = v,v' = a_s,\tau' = 1\; \& \;  v \ge 0 \land \tau \le \Tau  &&
    \end{flalign}
    \begin{flalign}\label{eq:m1:4}
            \ctrl \triangleq \opif \, (\safe) \; a_s\coloneqq a_n \; \opelse \; a_s\coloneqq -\aMinS &&
    \end{flalign} 
    \begin{flalign}\label{eq:m1:5}
        \safe \triangleq x_c - x \ge \msdt &&
    \end{flalign}
    \begin{flalign}\label{eq:m1:6}
         \msd \triangleq vt + \frac{\aMaxN\,t^2}{2} + \frac{\left(v+\aMaxN\,t\right)^2}{2\aMinS} && 
    \end{flalign}
    \begin{flalign}
        \env \triangleq &\:x_c\coloneqq*;\,?\,\Big(x_c-x \ge \msdi \Big);& \label{eq:m1:7}\\ 
                          &a_n\coloneqq *;\; ?\left(-\aMinN \le a_n \le \aMaxN\right) &\label{eq:m1:7a}
    \end{flalign}
    \begin{flalign}\label{eq:m1:8}
        \assume \triangleq &\:\aMinS > 0 \land \aMaxN > 0 \land \aMinN > 0 \land \Tau>0 &\nonumber\\ 
                          &\land v\ge 0 \land x_c-x \ge \msdi &
    \end{flalign}
  \end{HP}
\end{figure}

The distance traveled (change in position \pos) and the change in velocity \vel due to a constant acceleration \acc during the time interval \delt can be computed from the solution to the differential equations~\eqref{eq:vehicleode} as (for initial values $x_0$ and $v_0$): 
\begin{align}
    & x\,(t) = x_0 + v_0t + \frac{a t^2}{2} \label{eq:solx} & \\ 
    & v\,(t) = v_0 + a t \label{eq:solv} & 
\end{align}

The condition \safe~\eqref{eq:m1:5} checks whether the distance between the critical position $x_c$ and current position \pos is at least \msdt, i.e., the minimal safe distance for one execution cycle $\Tau$, the maximum time interval between two decisions allowed by the model. 

The HP \env models the behavior of the Situation Assessment (gives $\langle x_c, v_c \rangle$) and the Nominal Controller (gives $a_n$). Since $v_c = 0$ in this simplified case, we only consider $x_c$ in \env. While it is desirable to prove that the decision-making~\eqref{eq:m1:5} fulfills \guarantee for all possible values of $x_c$, it follows from Fig.~\ref{fig:ill1} and the discussion in Section~\ref{sec:intro} that such a proof cannot be obtained. For instance, no controller can guarantee safety from an initial state that already violates the safety constraint. However, we should be able to fulfill \guarantee in all behaviors where it is practically feasible to act in a safe manner. The formula~\eqref{eq:m1:7} describes this intuition with a nondeterministic assignment followed by a test action. The test $?\,\Big(x_c-x \ge \msdi \Big)$ in~\eqref{eq:m1:7} only admits behaviors where the distance between $x_c$ and $x$ is at least the minimal safety distance for zero duration \msdi, i.e., the Safety Controller can fulfill the safety constraint by maximal braking from the current state. The inequality $x_c - x \ge \msdi$ characterizes the admissible region (see Fig.~\ref{fig:ill1}) for this model. Similarly,~\eqref{eq:m1:7a} describes the behavior of the Nominal Controller, where it can nondeterministically output any value within the bounds, i.e., $\left(-\aMinN \le a_n \le \aMaxN\right)$.

The formula \assume~\eqref{eq:m1:8} specifies the initial conditions for Model~\ref{alg:sm1}; the four symbolic system parameters are positive, the velocity $v$ is non-negative, and that the system starts within the admissible region $x_c - x \ge \msdi$. Fig.~\ref{fig:sim2} shows a simulation of the \ego's motion with Model~\ref{alg:sm1}.

\begin{figure}[thpb] 
\centering
\includegraphics[width=0.9\linewidth,keepaspectratio]{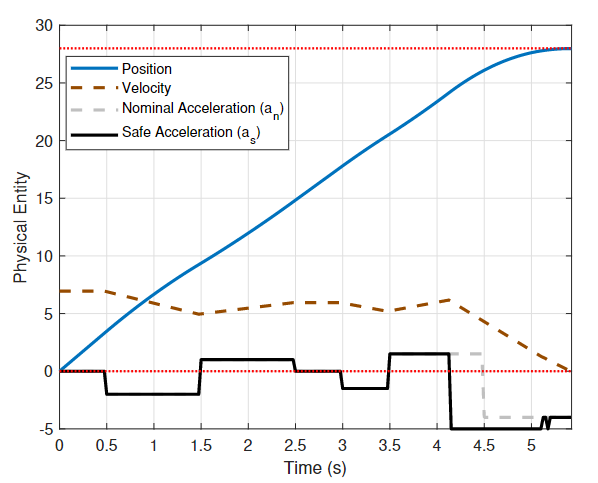}
\caption{An example simulation of ego-vehicle's motion with Model~\ref{alg:sm1}. $x_c = \SI{28}{\metre}$ and $v_c = \SI{0}{\metre/\second}$ indicated by the red dotted line. $a_s$ follows $a_n$ until $\SI{4.15}{\second}$ at which point the safety controller decides $a_n$ is no longer safe and brakes with $a_s=-\aMinS=\SI{-5}{\metre/\second^2}$.}
\label{fig:sim2}
\end{figure}

To prove the \dl formula~\eqref{eq:m1:1}, we use the interactive theorem prover \keyx, which takes a \dl formula as an input and proves it by successively decomposing it into several sub-goals according to the proof rules of \dl~\cite{platzer2012logics,platzer2018logical}. To handle loops, \keyx uses invariants to inductively reason about all (non-aborting) executions of the loop through the loop invariant proof rule~\cite{platzer2018logical}. The loop invariant rule uses some (inductive) loop invariant $\upzeta$ to prove~\eqref{eq:m1:1} by proving three separate formulas: \begin{enumerate}[label=(\roman*)]
    \item $(\assume) \to \upzeta$
    \item $\upzeta \to [\env;\,\ctrl;\,\plant]\,\upzeta$
    \item $\upzeta \to (\guarantee)$.
\end{enumerate}

\begin{theorem}\label{th:1}
Model $\mdl_1$ for the Decision \& Control module described by~\eqref{eq:m1:3}-\eqref{eq:m1:7a} guarantees to provide a safe acceleration request with respect to the safety constraint $\langle x_c,v_c\rangle$, with $v_c=0$, as expressed by the \dl formula~\eqref{eq:m1:1}.
\end{theorem}

\begin{proof}
Theorem~\ref{th:1} is proved in~\keyx. The proof uses the loop invariant $\upzeta \equiv x_c-x \ge \msdi$~\cite{selvarajgithub}.
\end{proof}

\begin{remark}
Though only $v_c = 0$ is considered in Model~\ref{alg:sm1}, the use of nondeterminism in the model enables a safety proof that covers a wide variety of designs. For instance, the nondeterministic assignments in~\eqref{eq:m1:7} and~\eqref{eq:m1:7a} make the proof independent of the implementation of the Situation Assessment and the Nominal Controller (Fig.~\ref{fig:funcarch}). Furthermore, symbolic bounds on the system parameters make the proof cover infinitely many design variants. 
\end{remark}

\subsection{Model 2: Conservative with \texorpdfstring{$v_c \geq 0$}{vCrit larger or equal to zero}}\label{sec:sm:gr}

Model~\ref{alg:sm2} extends Model~\ref{alg:sm1} to allow $v_c \geq 0$. This change is described by modifications to the \guarantee~\eqref{eq:m2:2}, \env~\eqref{eq:m2:7}, and \assume~\eqref{eq:m2:8} to include $v_c$. Similarly, the formula for the minimal safety distance $\msd$~\eqref{eq:m2:6} is adjusted to reflect the generic case while following the same worst-case reasoning as Model~\ref{alg:sm1}. Here, \msd is given by the sum of the distance traveled from the current state by accelerating with $\aMaxN$ for time interval $t$ and from there on, the distance traveled to reach $v_c \geq 0$ (instead of $v_c=0$ as in Model~\ref{alg:sm1}). 

\begin{theorem}\label{th:2}
Model $\mdl_2$ for the Decision \& Control module, described by~\eqref{eq:m2:3}-\eqref{eq:m2:7a} guarantees to provide a safe acceleration request with respect to the safety constraint $\langle x_c,v_c\rangle$ as expressed by the \dl formula~\eqref{eq:m2:1}.
\end{theorem}

\begin{proof}
Theorem~\ref{th:2} is proved in~\keyx using the loop invariant $\upzeta \equiv v_c\ge 0\,\land \, x_c-x \ge \msdi$~\cite{selvarajgithub}.
\end{proof}

\begin{figure}[!t]
 \removelatexerror
  \begin{HP}[H]
    \DontPrintSemicolon
    \setstretch{0.5}
    \caption{Conservative with $v_c\geq 0$}
    \label{alg:sm2}
    \SetAlgoLined
    
    \begin{flalign} \label{eq:m2:1}
        (\assume) \to [(\env;\,\ctrl;\,\plant)^*]\,(\guarantee) &&
    \end{flalign}
    \begin{flalign}\label{eq:m2:2}
        \guarantee \triangleq (x \ge x_c \to v \le v_c) &&
    \end{flalign}
    \begin{flalign}\label{eq:m2:3}
        \plant \triangleq \tau\coloneqq 0;\, x' = v,v' = a_s,\tau' = 1\; \& \;  v \ge 0 \land \tau \le \Tau  &&
    \end{flalign}
    \begin{flalign}\label{eq:m2:4}
            \ctrl \triangleq \opif \, (\safe) \; a_s\coloneqq a_n \; \opelse \; a_s\coloneqq -\aMinS &&
    \end{flalign} 
    \begin{flalign}\label{eq:m2:5}
        \safe \triangleq x_c - x \ge \msdt &&
    \end{flalign}
    \begin{flalign}\label{eq:m2:6}
         \msd \triangleq vt + \frac{\aMaxN t^2}{2} + \frac{\left(v+\aMaxN t\right)^2-v_c^2}{2\aMinS} && 
    \end{flalign}
    \begin{flalign}
        \env \triangleq &\:x_c\coloneqq*;\,v_c\coloneqq*; &\nonumber \\
                       &?\,\Big(v_c\ge 0 \land x_c-x \ge \msdi \Big);&\label{eq:m2:7}\\ 
                       &a_n\coloneqq *;\; ?\left(-\aMinN \le a_n \le \aMaxN\right) &\label{eq:m2:7a}
    \end{flalign}
    \begin{flalign}\label{eq:m2:8}
        \assume \triangleq &\:\aMinS > 0 \land \aMaxN > 0 \land \aMinN > 0 \land \Tau>0 &\nonumber\\ 
                          &\land v\ge 0 \land v_c \ge 0 \land x_c-x \ge \msdi &
    \end{flalign}    
  \end{HP}
\end{figure}

\subsection{Model 3: Permissive with \texorpdfstring{$v_c=0$}{vCrit zero}}\label{sec:sm:sp}

The decision-making in the two models discussed so far is based on a threat metric $\msd$ determined from worst-case assumption of the \ego behavior. While provably safe, such worst-case reasoning often leads to a conservative design. The \emph{permissive} models described by Model~\ref{alg:sm3} and Model~\ref{alg:sm4} are based on a threat metric that relaxes the worst-case assumption. This section discusses the case where $v_c = 0$ and Section~\ref{sec:sm:gp} extends it for the generic case $v_c\geq 0$. 

\begin{figure}[!t]
 \removelatexerror
  \begin{HP}[H]
    \DontPrintSemicolon
    \setstretch{0.5}
    \caption{Permissive with $v_c = 0$}    \label{alg:sm3}
    \SetAlgoLined
    
    \begin{flalign} \label{eq:m3:1}
        (\assume) \to [(env;\,ctrl;\,plant)^*]\,(guarantee) &&
    \end{flalign}
    \begin{flalign}\label{eq:m3:2}
        guarantee \triangleq (x \ge x_c \to v = 0) &&
    \end{flalign}
    \begin{flalign}\label{eq:m3:3}
        plant \triangleq \tau\coloneqq 0;\, x' = v,v' = a_s,\tau' = 1\; \& \;  v \ge 0 \land \tau \le \Tau  &&
    \end{flalign}    
    \begin{flalign}\label{eq:m3:4}
            \ctrl \triangleq \opif \, (\safe) \; a_s\coloneqq a_n \; \opelse \; a_s\coloneqq -\aMinS &&
    \end{flalign} 
    \begin{flalign}\label{eq:m3:5}
        \safe \triangleq x_c - x \ge \msdt &&
    \end{flalign}
    \begin{flalign}\label{eq:m3:6}
         \msd \triangleq  &
         \begin{cases}
                                vt + \dfrac{a_n t^2}{2} + \dfrac{\left(v+a_n t\right)^2}{2\aMinS} & \text{if $v+a_n t\ge 0$}\\
                                -\dfrac{v^2}{2a_n} & \text{otherwise}
                            \end{cases}&
    \end{flalign}
    \begin{flalign}
        env \triangleq &\:x_c\coloneqq*;\,?\,\Big(x_c-x \ge \msdi \Big);&\label{eq:m3:7}\\ 
                          &a_n\coloneqq *;\; ?\left(-\aMinN \le a_n \le \aMaxN\right) &\label{eq:m3:7a}
    \end{flalign}
    \begin{flalign}\label{eq:m3:7b}
        \msdi \triangleq \frac{v^2}{2\aMinS} &&
    \end{flalign}       
    \begin{flalign}\label{eq:m3:8}
        \assume \triangleq &\:\aMinS > 0 \land \aMaxN > 0 \land \aMinN > 0 \land \Tau>0 &\nonumber\\ 
                          &\land v\ge 0 \land x_c-x \ge \msdi &
    \end{flalign}    
  \end{HP}
\end{figure}

In comparison to Model~\ref{alg:sm1}, the threat metric $\msd$ used to assess $a_n$ is changed according to the new relaxed assumption. Intuitively, while calculating $\msd$, the worst-case behavior of accelerating with $\aMaxN$ for time interval $t$ is replaced with the actual behavior of accelerating with the requested $a_n$ for the time interval $t$. This change is reflected in Model~\ref{alg:sm3} by replacing $\aMaxN$ with $a_n$ in~\eqref{eq:m1:6} to give:
\begin{equation}\label{eq:wrongmst}
    \msd \triangleq vt + \frac{a_n t^2}{2} + \frac{\left(v+a_n t\right)^2}{2\aMinS}
\end{equation}

The Safety Controller decides whether $a_n$ is \safe by checking if $x_c - x \ge \msdt$ is true. From~\eqref{eq:wrongmst}, note that \msdt is given by the sum of the distance traveled from the current state due to $a_n$ for time interval $\Tau$, and from there on the distance traveled by braking with $\aMinS$ until zero velocity is reached. Since $-\aMinN \le a_n \le \aMaxN$, the \ego can either accelerate or brake during a given time interval depending on the requested $a_n$. However, when the \ego brakes $(-\aMinN \le a_n < 0)$, it is possible that zero velocity is reached during interval $\Tau$ but not necessarily after $\Tau$ time. In such edge cases, \msdt determined from \eqref{eq:wrongmst} will lead to incorrect and unsafe decisions and therefore is not sufficient to fulfill the \guarantee~\eqref{eq:m3:2}.

An example simulation of such an edge case is given in Fig.~\ref{fig:sim3}. Though the safety constraint is violated during time interval $\Tau$ for some requested $a_n < 0$, the test condition $x_c - x \ge \msdt$ incorrectly decides the requested $a_n$ to be \safe using \msdt determined from~\eqref{eq:wrongmst}. To account for such scenarios, $\msd$ is split into two cases depending on whether the velocity after time interval $t$ due to $a_n$ is non-negative or not, as formulated in~\eqref{eq:m3:6}. In the second case, when $(v+a_n t) < 0$ due to $-\aMinN \le a_n < 0$, a choice of $a_n$ is considered \safe if braking with $a_n$ to a complete stop is sufficient to satisfy the constraint, as described by $x_c - x \ge -\frac{v^2}{2a_n}$. Else, $\aMinS$ is set as the safe acceleration value. 

\begin{figure}[thpb] 
\centering
\includegraphics[width=0.9\linewidth,keepaspectratio]{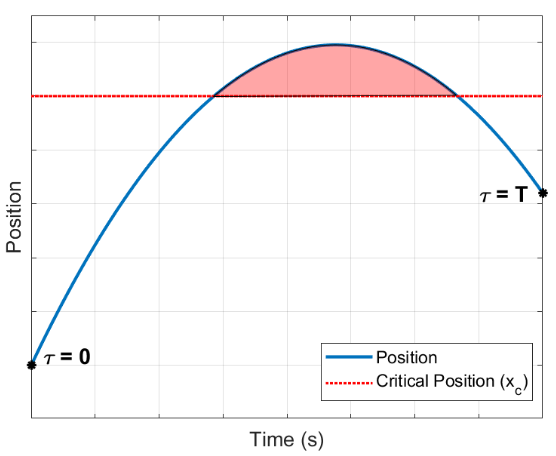}
\caption{An illustration of incorrect decision-making. At $\tau=0$, a decision function using~\eqref{eq:wrongmst} incorrectly decides the choice of $a_n$ to be safe. Though the constraint is not violated at $\tau=\Tau$, it is violated in the shaded region before $\tau=\Tau$. Such edge cases are resolved by the case-split in~\eqref{eq:m3:6}.}
\label{fig:sim3}
\end{figure}

\begin{theorem}\label{th:3}
Model $\mdl_3$ for the Decision \& Control module, described by~\eqref{eq:m3:3}-\eqref{eq:m3:7b} guarantees to provide a safe acceleration request with respect to the safety constraint $\langle x_c,v_c\rangle$, with $v_c=0$, as expressed by the \dl formula~\eqref{eq:m3:1}.
\end{theorem}

\begin{proof}
Theorem~\ref{th:3} is proved in~\keyx using the loop invariant $\upzeta \equiv x_c-x \ge \msdi$~\cite{selvarajgithub}.
\end{proof}

\begin{remark}
The edge case encountered in Model~\ref{alg:sm3} highlights the need for reasoning about intermediate states to accurately prove safety \emph{throughout all} executions of the model. In the \plant description~\eqref{eq:m3:3}, the implicit nondeterminism introduced by the domain constraint $\tau \le \Tau$ allows evolution of any duration $\tau$ (including zero) that satisfies the domain constraint. Therefore, Theorem~\ref{th:3} shows that \guarantee is fulfilled throughout all possible executions of the model. However, modifying the \plant~\eqref{eq:m3:3} to include a test as:
\begin{equation}\label{eq:wrongplant}
    \tau\coloneqq 0; x' = v,v' = a_s,\tau' = 1\,\& \,  v \ge 0 \land \tau \le \Tau; ?(\tau=\Tau);
\end{equation}
makes the \dl formula~\eqref{eq:m3:1} provable even with the incorrect \msd~\eqref{eq:wrongmst}, since~\eqref{eq:wrongplant} allows only evolution of exactly $\tau = \Tau$ duration in the model and therefore requires \guarantee to hold in only those states reached at the end of evolution for precisely $\Tau$ duration. 
\end{remark}

\subsection{Model 4: Permissive with \texorpdfstring{$v_c\geq 0$}{vCrit greater or equal to zero}}\label{sec:sm:gp}
The permissive model where $v_c \ge 0$ is similar to Model~\ref{alg:sm2} in Section~\ref{sec:sm:gr} except for the threat metric $\msd$~\eqref{eq:m4:6}. Here, since \msd is determined based on actual behavior and not worst-case behavior of the \ego, it is split into two cases to account for edge cases described in the previous section and illustrated in Fig.~\ref{fig:sim3}.

\begin{theorem}\label{th:4}
Model $\mdl_4$ for the Decision \& Control module, described by~\eqref{eq:m4:3}-\eqref{eq:m4:9} guarantees to provide a safe acceleration request with respect to the safety constraint $\langle x_c,v_c\rangle$ as expressed by the \dl formula~\eqref{eq:m4:1}.
\end{theorem}

\begin{proof}
Theorem~\ref{th:4} is proved in~\keyx using the loop invariant $\upzeta \equiv v_c\ge 0\,\land \, x_c-x \ge \msdi$~\cite{selvarajgithub}.
\end{proof}

\begin{figure}[!t]
 \removelatexerror
  \begin{HP}[H]
    \DontPrintSemicolon
    \setstretch{0.5}
    \caption{Permissive with $v_c\geq 0$}
    \label{alg:sm4}
    \SetAlgoLined
    
    \begin{flalign} \label{eq:m4:1}
        (\assume) \to [(env;\,ctrl;\,plant)^*]\,(guarantee) &&
    \end{flalign}
    \begin{flalign}\label{eq:m4:2}
        guarantee \triangleq (x \ge x_c \to v \le v_c) &&
    \end{flalign}
    \begin{flalign}\label{eq:m4:3}
        plant \triangleq \tau\coloneqq 0;\, x' = v,v' = a_s,\tau' = 1\; \& \;  v \ge 0 \land \tau \le \Tau  &&
    \end{flalign}
    \begin{flalign}\label{eq:m4:4}
        \ctrl \triangleq \opif \, (\safe) \; a_s\coloneqq a_n \; \opelse \; a_s\coloneqq -\aMinS &&
    \end{flalign} 
    \begin{flalign}\label{eq:m4:5}
        \safe \triangleq x_c - x \ge \msdt &&
    \end{flalign}
    \begin{flalign}\label{eq:m4:6}
         \msd \triangleq  &\begin{cases}
                                vt + \dfrac{a_n t^2}{2} + \dfrac{\left(v+a_n t\right)^2 - v_c^2}{2\aMinS} & \text{if $v+a_n t\ge 0$}\\[5pt]
                                -\dfrac{v^2}{2a_n} & \text{otherwise}
                            \end{cases}&
    \end{flalign}
    \begin{flalign}
        env \triangleq &\:x_c\coloneqq*;\,v_c\coloneqq*; &\nonumber \\
                       &?\,\Big(v_c\ge 0 \land x_c-x \ge \msdi \Big);&\label{eq:m4:8}\\
                       &a_n\coloneqq *;\; ?\left(-\aMinN \le a_n \le \aMaxN\right) &\label{eq:m4:8a}
    \end{flalign}
    \begin{flalign}\label{eq:m4:9}
        \msdi \triangleq \frac{v^2-v_c^2}{2\aMinS} &&
    \end{flalign}    
    \begin{flalign}\label{eq:m4:10}
        \assume \triangleq &\:\aMinS > 0 \land \aMaxN > 0 \land \aMinN > 0 \land \Tau>0 &\nonumber\\ 
                          &\land v\ge 0 \land v_c \ge 0 \land x_c-x \ge \msdi &
    \end{flalign}    
  \end{HP}
\end{figure}

\subsection{Model 5: Required Acceleration as Threat Metric}\label{sec:sm:5}
In all the previous models, the decision-making is based on minimal safety distance, a threat metric defined in the distance domain. However, it is possible to define threat metrics in other domains like time, acceleration, etc.~\cite{jansson2005collision}. In this section, we show how Model~\ref{alg:sm3} can be reformulated with a threat metric in the acceleration domain \aReq defined as the \emph{longitudinal acceleration required} to fulfill the safety constraint. The decision-making~\eqref{eq:m5:4} is similar to the other models where the \safe condition~\eqref{eq:m5:5} is used to assess $a_n$ by comparing \aReqT to a threshold \ath.

Since the intention here is to have a reformulation of the threat metric in the acceleration domain, \aReqT in~\eqref{eq:m5:6} is defined similar to~\eqref{eq:m3:6} in Model~\ref{alg:sm3} using the same assumption for the \ego  behavior. However, the decision-making in Model~\ref{alg:sm3} and Model~\ref{alg:sm5} differ in the threshold used for the Safety Controller's intervention and this difference is further discussed in Section~\ref{sec:performance}. Furthermore, while no assumption on the relation between \aMinN and \aMinS was required in Models~\ref{alg:sm1}-\ref{alg:sm4}, \assume~\eqref{eq:m5:10} requires $\aMinN < \aMinS$ (or $\aMinN \le \aMinS$) to fulfill \guarantee~\eqref{eq:m5:2} as $\aMinN$ is used as a threshold for comparison in the \safe condition~\eqref{eq:m5:5}.

\begin{figure}[!t]
 \removelatexerror
  \begin{HP}[H]
    \DontPrintSemicolon
    \setstretch{0.5}
    \caption{\aReq as threat metric}
    \label{alg:sm5}
    \SetAlgoLined
    
    \begin{flalign} \label{eq:m5:1}
        (\assume) \to [(env;\,ctrl;\,plant)^*]\,(guarantee) &&
    \end{flalign}
    \begin{flalign}\label{eq:m5:2}
        guarantee \triangleq (x \ge x_c \to v = 0) &&
    \end{flalign}
    \begin{flalign}\label{eq:m5:3}
        plant \triangleq \tau\coloneqq 0;\, x' = v,v' = a_s,\tau' = 1\; \& \;  v \ge 0 \land \tau \le \Tau  &&
    \end{flalign}    
    \begin{flalign}\label{eq:m5:4}
        \ctrl \triangleq \opif \, (\safe) \; a_s\coloneqq a_n \; \opelse \; a_s\coloneqq -\aMinS &&
    \end{flalign} 
    \begin{flalign}\label{eq:m5:5}
        \safe \triangleq \aReqT \ge \ath &&
    \end{flalign}
    \begin{flalign}\label{eq:m5:6}
         \aReqT \triangleq  &\begin{cases} 
                                -\dfrac{\left(v + a_n\Tau\right)^2}{2\left(x_c-x-v\Tau-\dfrac{a_n\Tau^2}{2}\right)} & \text{if $v+a_n\Tau \ge 0$}\\[8pt]
                                -\dfrac{v^2}{2(x_c-x)} & \text{otherwise}
                            \end{cases}&
    \end{flalign}
    \begin{flalign}\label{eq:m5:7}
         \ath \triangleq  &\begin{cases}
                                -\aMinN & \text{if $v+a_n\Tau \ge 0$}\\
                                -a_n &    \text{otherwise}           
                            \end{cases}&
    \end{flalign}
    \begin{flalign}\label{eq:m5:8}
        env \triangleq &\:x_c\coloneqq*;\,?\,\Big(\aReqi \ge -\aMinS \Big);&\\ 
                          &a_n\coloneqq *;\; ?\left(-\aMinN \le a_n \le \aMaxN\right) &
    \end{flalign}
    \begin{flalign}\label{eq:m5:9}
        \aReqi \triangleq -\frac{v^2}{2(x_c-x)} &&
    \end{flalign}   
    \begin{flalign}\label{eq:m5:10}
        \assume \triangleq &\:\aMinS > 0 \land \aMaxN > 0 \land \aMinN > 0 \land \Tau>0 &\nonumber\\ 
                          &\land v\ge 0 \land \aMinN < \aMinS \land \aReqi \ge -\aMinS &
    \end{flalign}
  \end{HP}
\end{figure}

\begin{theorem}\label{th:5}
Model $\mdl_5$ for the Decision \& Control module, described by~\eqref{eq:m5:3}-\eqref{eq:m5:9} guarantees to provide a safe acceleration request with respect to the safety constraint $\langle x_c,v_c\rangle$, with $v_c=0$, as expressed by the \dl formula~\eqref{eq:m5:1}.
\end{theorem}

\begin{proof}
Theorem~\ref{th:5} is proved in~\keyx using the loop invariant $\upzeta \equiv \frac{v^2}{2(x_c - x)} \ge -\aMinS$~\cite{selvarajgithub}.
\end{proof}

\begin{remark}
Note that the formulation~\eqref{eq:m5:6} can be extended to describe yet another threat metric like \emph{Brake Threat Number}~\cite{btn}, defined as the ratio of the longitudinal acceleration required to the maximum longitudinal acceleration. Though the decision function in all the models are developed based on solutions to~\eqref{eq:vehicleode} and assumptions on the safe behavior of the \ego, the models differ in the decision-making when it comes to the Safety Controller's interventions and this section highlights how different design variants can be modeled as hybrid programs. The main effort in the verification process consists in identifying appropriate assumptions and loop invariants.
\end{remark}

\section{MODEL PERFORMANCE ANALYSIS} \label{sec:performance}
In Section~\ref{sec:models}, theorems~\ref{th:1}, \ref{alg:sm3}, and \ref{th:5} verified that the corresponding models provide a safe acceleration request with respect to $\langle x_c, v_c \rangle$ where $v_c = 0$. Though all three models are proved safe, they differ in their decision-making and hence have different performance. In all the models, the Safety Controller, using a threat metric, assesses whether the nominal acceleration $a_n$ compromises safety and if so, intervenes with maximal braking. Consequently, the performance depends on when and how the intervention is made. While an overly conservative Safety Controller that intervenes often with maximal braking can guarantee safety, it also possesses a risk of limited user acceptance. Therefore, it is valuable to analyze the performance of the Safety Controller to choose a good  design. 

One way to conduct such an analysis is to simulate the models in different sets of scenarios (e.g. Fig~\ref{fig:sim2}) and compare them with a suitable performance metric. The shortcoming with such an approach is the intractability to compare \emph{all} possible scenarios. An alternative approach is to obtain a formal machine checked proof about the relation between different models for all parametric combinations. For instance, the condition \safe~\eqref{eq:m1:5}, \eqref{eq:m3:5}, and \eqref{eq:m5:5} determines when the Safety Controller intervenes in models~\ref{alg:sm1},~\ref{alg:sm3}, and~\ref{alg:sm5}, respectively. A Safety Controller that intervenes as late as possible and still guarantees safety is certainly a preferable choice for a good performance. Therefore, an obvious choice is to use the minimal safety distance to compare the models.    

In Model~\ref{alg:sm1}, the minimal safety distance for intervention is given by 
\begin{equation}
    \mone \triangleq v\Tau + \frac{\aMaxN\,\Tau^2}{2} + \frac{\left(v+\aMaxN\,\Tau\right)^2}{2\aMinS}
\end{equation}
and depends on the velocity \vel and the system parameters \aMaxN, \aMinS, and $\Tau$. Similarly, for Model~\ref{alg:sm3} and Model~\ref{alg:sm5}, \mthree and \mfive can be obtained from \eqref{eq:m3:6} and \eqref{eq:m5:6} respectively. Though~\eqref{eq:m5:6} uses \aReq as the threat metric, it is straightforward to reformulate it to obtain a safety distance as discussed in Section~\ref{sec:sm:5}. With the minimal safety distances for intervention obtained, the relation between the models for all parametric combinations can be captured by the first-order logic formula of real arithmetic:
\begin{align}
    \forall\begin{pmatrix} \aMinS \\ \aMinN \\ \aMaxN \\ \Tau \\ \vel \\ a_n \end{pmatrix}
    \Biggl( \left(\begin{bmatrix} 0 \\ 0 \\ 0 \\ \aMinN          \end{bmatrix} <
    \begin{bmatrix} \aMinN \\ \aMaxN \\ \Tau \\ \aMinS \end{bmatrix} \land 
    0 \le v \land -\aMinN \le a_n \le \aMaxN
    \right) \to \nonumber \\
    (\mone \ge \mthree \land \mfive \ge \mthree)\Biggr)\label{eq:modelcompare}
\end{align}
\begin{theorem}\label{th:6}
The Safety Controller in Model~\ref{alg:sm3} uses a smaller safety distance for intervention compared to the Safety Controllers in Model~\ref{alg:sm1} and Model~\ref{alg:sm5} as expressed by~\eqref{eq:modelcompare}.
\end{theorem}
\begin{proof}
Theorem~\ref{th:6} is proved in~\keyx~\cite{selvarajgithub}. 
\end{proof}
\begin{remark}
Note that no loop invariant was necessary to prove Theorem~\ref{th:6} since we only reason about a first-order logic formula of real arithmetic without a hybrid program.
\end{remark}

From Theorem~\ref{th:6}, we can conclude that the Safety Controller in Model~\ref{alg:sm3} does not intervene earlier than the controllers in Model~\ref{alg:sm1} and Model~\ref{alg:sm5} to guarantee safety and therefore performs better. It is indeed possible to arrive at the same conclusion by analytically deriving the relation between the different safety distances. Of course, such a manual approach does not scale in practice. However, obtaining a formal machine checked proof as discussed in this section scales well to reason about different designs and to compare them using different performance metrics.  

In all the models discussed so far, the Safety Controller intervenes with maximal braking, $a_s\coloneqq -\aMinS$. Though proven safe, it might not always be necessary to intervene with maximal braking in order to satisfy the constraint. Certainly, it is possible to change how the intervention is made in the models. For instance, if $\aReq$ is the acceleration required to satisfy the constraint at any given point, a modification to $\ctrl$ as shown in~\eqref{eq:ctrlnew} describes that, as proved by~\keyx, the Safety Controller can intervene with any acceleration value that is bounded by \aReq and \aMinS:
\begin{equation}\label{eq:ctrlnew}
    \ctrl \triangleq \opif \, (\safe) \; a_s\coloneqq a_n \; \opelse \;a_s \coloneqq *;\,?(-\aMinS \le a_s \land a_s \le \aReq)
\end{equation}
Clearly, modeling the intervention with nondeterminism as in~\eqref{eq:ctrlnew} covers different variations of controller implementation with the same proof. Section~\ref{sec:models} described how \dl is used to formally analyze the Decision \& Control module for an in-lane AD feature, which was the primary goal. In this section, we have shown how \keyx can be used to compare the verified models to aid in the design and development of the AD feature.   

\section{REQUIREMENT REFINEMENT} \label{sec:identification}
In this section, we discuss how \dl and in particular the nondeterministic operators can be used to refine the requirements for the components interacting with the Decision \& Control module. Theorems~\ref{th:1}-\ref{th:5} of Section~\ref{sec:models} verified that the respective models provide a safe acceleration request with respect to the safety constraint. The nondeterminism in the models verifies a wide range of concrete implementations. However, certain assumptions are included in the models to prove the \guarantee in each case. Broadly, the assumptions in the models are described in three ways:
\begin{enumerate}[label=(\roman*)]
    \item bounds on the system parameters, 
    \item assumptions on the interacting component/environment behavior, and 
    \item evolution domain constraint.
\end{enumerate}
These assumptions have resulted from the formal analysis of the decision-making in the Safety Controller and here, we show how such insights are used to formalize safety requirements for all the components in the Decision \& Control module and also help to identify relevant ODD conditions.  

Consider the functional safety requirement, FSR\,\ref{fsr:dc} in Section~\ref{sec:example} for the Decision \& Control module. Of course, FSR\,\ref{fsr:dc} can be formulated using the safety constraint pair, $\langle x_c, v_c \rangle$ and thus Theorem~\ref{th:1} verifies FSR\,\ref{fsr:dc} for Model~\ref{alg:sm1} in Section~\ref{sec:sm:sr}. Specifically in Model~\ref{alg:sm1}, the formulas \assume~\eqref{eq:m1:8} and \plant~\eqref{eq:m1:3} include the assumptions on the system parameters; \env~\eqref{eq:m1:7}--\eqref{eq:m1:7a} and \plant~\eqref{eq:m1:3} describe the assumptions on the components interacting with the Safety Controller; and finally the \plant~\eqref{eq:m1:3} includes the evolution domain constraint. A straightforward consequence of the domain constraint $v\ge0$ is that the evolution of the \plant would stop before reaching negative velocity, thus the \ego does not travel backwards. Naturally, such constraints can be used to refine the ODD for the AD feature as the safety guarantee clearly does not hold in situations where the \ego might possibly travel backwards, e.g. road geometries with high slope and less friction. Furthermore, the assumptions on the interacting components obtained from Model~\ref{alg:sm1} can be used to further refine FSR\,\ref{fsr:dc} as:
\begin{newfsr}{1.1}\label{fsr:1:1}
The Nominal Controller shall output a nominal acceleration $a_n$ such that $(-\aMinN \le a_n \le \aMaxN)$.
\end{newfsr}
\begin{newfsr}{1.2}\label{fsr:1:2}
The Situation Assessment shall output a critical position $x_c$ for a given \ego position $x$, velocity $v$, and maximum braking capability $\aMinS$ such that $x_c - x \ge \frac{v^2}{2\aMinS}$. 
\end{newfsr}
\begin{newfsr}{1.3}\label{fsr:1:3}
The Safety Controller shall at all times output a safe acceleration value $a_s$ to avoid a collision with any object in front if FSR~\ref{fsr:1:1} and FSR~\ref{fsr:1:2} are met.  
\end{newfsr}
\begin{newfsr}{1.4}\label{fsr:1:4}
The Vehicle Control shall always control the \ego according to the safe acceleration request $a_s$ to avoid a collision with any object in front.  
\end{newfsr}

The safety requirements thus obtained can be used in the subsequent analysis of the respective components using HPs and \dl. For instance, one conceivable but na\"ive algorithm for the Situation Assessment component is to sort the objects in front of the \ego and select the position of the closest object $x_l$ to obtain the critical position as 
\begin{equation}
    sa \triangleq x_c \coloneqq x_l + d \label{eq:cg:1} 
\end{equation}
where parameter $d$ denotes an admissible separation between the \ego and the object in front when the \ego is completely stopped. While this is a very simple description, it follows from a worst-case reasoning for object behavior (e.g. a leading vehicle coming to an immediate stop anytime). In this case, verifying FSR~\ref{fsr:1:2} can be translated into proving the \dl formula:
\begin{equation}\label{eq:sa:dl}
    (\assume) \to [(sense;\,sa;\,\ctrl;\,\plant)^*]\,\left(x_c - x \ge \frac{v^2}{2\aMinS}\right)
\end{equation}
where \ctrl and \plant model the dynamics of both the \ego and the object in the environment. Verifying~\eqref{eq:sa:dl} using \keyx requires the identification of necessary assumptions on the Sense module, which can subsequently be used to obtain requirements on the sensor range for the AD feature and/or refine the ODD conditions. Modeling the Situation Assessment in a modular way as described by~\eqref{eq:sa:dl} also gives the flexibility to easily reason about various algorithmic variants similar to how different threat metrics are handled in the models for the Safety Controller in Section~\ref{sec:models}. For example, to relax the worst-case reasoning for leading vehicles in~\eqref{eq:cg:1} from an immediate stop to braking with at least $B$ from velocity $v_l$ to a stop, then \emph{sa} in~\eqref{eq:sa:dl} can be replaced by:
\begin{equation}\label{eq:sa:2}
    sa \triangleq x_c \coloneqq x_l + \frac{v_l^2}{2B} + d. 
\end{equation}

Admittedly, the validity of the safety proofs in the deployed systems are highly dependent on the validity of the models, including the assumptions. For example, in Model~\ref{alg:sm1}, the \plant~\eqref{eq:m1:3} models the \ego behavior such that the safe acceleration request $a_s$ is accurately tracked by the Vehicle Control. However, it is often the case that actual deployed systems encounter disturbances, delays, etc., which makes it difficult to accurately track the request. Though \dl can model such disturbances and delays, the challenge manifests in identifying the corresponding invariant conditions to manage the complexity of the proofs to be constructed by the proof system.

\section{RELATED WORK}\label{sec:relwork}
Several approaches like testing, simulation, formal methods, etc. have been investigated to provide credible arguments that AD systems are safe and correct~\cite{koopman2019credible,riedmaier2020survey}. Formal methods, unlike approaches like testing and simulation, can exhaustively verify and guarantee the absence of errors through mathematical proofs of correctness of a model of the system. Formal analysis based on finite-state methods like supervisory control theory (SCT)~\cite{ramadge1989control}, or model checking~\cite{baier2008principles} have previously been used to reason about ADAS~\cite{korssen2017systematic}, AD~\cite{selvaraj2019verification,krook2019design,zita2017application} and other kinds of autonomous robotic systems~\cite{luckcuck2019formal}. Finite-state methods, though impressive in their own domains, limit the expressiveness of the models and require finite-state abstractions or approximations of the system. For example, formal analysis of traffic situations typically requires reasoning about continuous state variables like position, velocity, etc., that vary continuously with time, and obtaining finite-state abstractions of such entities risks an unconvincing argument of correctness. Furthermore, the use of exhaustive state-space exploration in SCT and model checking approaches is intractable for highly parametric AD systems. In comparison, the approach of this paper uses hybrid programs to model both continuous and discrete dynamics, and verifies them using mathematical proofs instead of exhaustive exploration, thereby covering infinitely many scenarios in each theorem and proof. Thus, a suitable trade-off between models closer to reality and a tractable formal analysis can be achieved. 

Another approach to formally verify the safety of AD  vehicles is through online reachability analysis~\cite{althoff2011set,onlinereachable}, where the verification is performed online by predicting the reachable sets from models of the AD vehicle and other traffic participants. A notable shortcoming in this approach is the heavy computational demand in the calculation of the reachable sets. Recent progress has been made in reducing the computational demand by making conservative model abstractions~\cite{gruber2018anytime} or by combining set-based reachability analysis with convex optimization~\cite{manzinger2020using}. In contrast, the verification approach used in our paper is completely offline and therefore does not contribute to the real-time computational demand in the AD vehicle. Moreover, the decision logic in the Safety Controller is---by design---intentionally simple in its behavior, thereby accommodating to the demands of the possibly complex Nominal Controller. Of course, the approach can be extended to include complex and more realistic models. However, a consequence of such realistic modeling is to deal with the proof complexity which might require additional manual effort to identify invariants and arithmetic simplifications to decide the validity of first-order formulas of real arithmetic.

Yet another approach to guarantee safety is to enforce set-invariance through control barrier functions as investigated in~\cite{ames2016control,lindemann2018control}. An important limitation in this method lies in the construction of such control barrier functions~\cite{li2021comparison}. In this regard, a similar problem with the deductive verification approach used in our paper is the identification of continuous invariants and loop invariants to improve proof automation~\cite{sogokon2019pegasus,platzer2010logical}. A comparison of the safety methods based on control barrier functions and reachability analysis is found in~\cite{li2021comparison}.

Differential dynamic logic (\dl) has been used in the specification and verification of adaptive cruise control~\cite{loos2011adaptive}, the European train control system~\cite{platzer2009european}, and aircraft collision avoidance~\cite{platzer2009formal}. The primary objective in those works is to demonstrate the application of \dl based verification in the respective case-studies. In comparison, in addition to showing how \dl is used in the safety argument of an AD feature, our paper discusses how such an approach can further aid in other development activities like comparison of the verified models and in the requirement refinement process. In~\cite{fulton2018safe,adelt2021formal}, \keyx in combination with runtime monitoring is used to guarantee safety of reinforcement learning-based controllers. Though our paper does not directly deal with reinforcement learning, as mentioned in Section~\ref{sec:prob}, the models and proofs presented can be used to guarantee the safety of any nominal control algorithm, including such based on reinforcement learning.   

This section presents a broad but inevitably incomplete overview of some research related to the use of formal methods to guarantee safety of automated systems. A more comprehensive survey on the formal specification and verification of autonomous robotic systems is found in~\cite{luckcuck2019formal}.

\section{CONCLUSION}\label{sec:conc}
The challenges in providing convincing arguments for the safe and correct behavior of AD systems is one obstacle their widespread commercial deployment. Formal methods and tools can help ensure the safety in various stages of AD development. This paper shows how differential dynamic logic (\dl) and the theorem prover \keyx can be used in the safety argument of an in-lane AD feature. Specifically,
\begin{enumerate}
    \item we have presented formal models and safety proofs of different design variants of a Decision \& Control module for an in-lane AD feature, 
    \item we have shown how the formal analysis helps in identifying the necessary assumptions and invariant conditions to guarantee safety, and 
    \item we have discussed and demonstrated how formal analysis using \dl and \keyx can not only be used to verify the different models but also aid in other development activities like refinement requirement and in evaluation of the different verified models.  
\end{enumerate}
Furthermore, the design and verification approach used in our paper identifies the necessary conditions on the interactions between the Safety Controller and the other components (e.g. Nominal Controller) to enforce safe behavior and therefore can be used to guarantee safety even if the Nominal Controller implements hard-to-verify (e.g., learning-based) algorithms. Though this paper only considers an in-lane AD feature, the approach can be extended to other types of AD features (e.g. lane changes). We believe that this work provides valuable insights for the use of formal methods in the safety argument of AD features. In the future we would like to refine the models such as introducing delays and disturbances, control decisions that combine steering and braking commands, etc., and investigate the proof effort required to guarantee their safety. A significant part of the verification effort depends on identifying the invariant conditions. In this paper, the loop invariants are identified manually. As part of future research, we will explore different methods to automatically identify such invariant candidates.

\bibliographystyle{IEEEtran}
\bibliography{IEEEabrv,mybib}

\newpage


\begin{IEEEbiography}[{\includegraphics[width=1in,height=1.25in,clip,keepaspectratio]{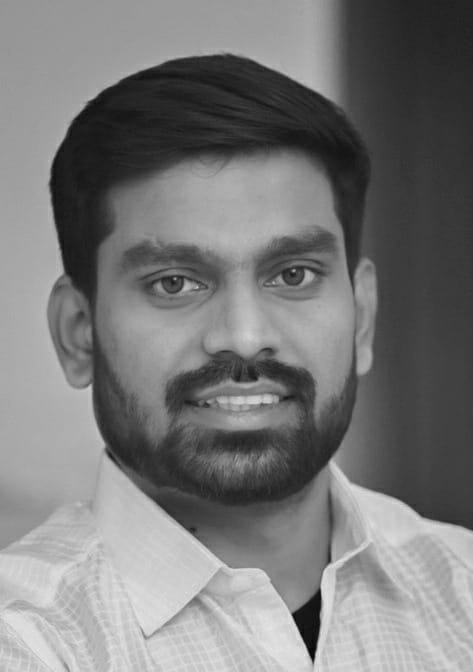}}]{Yuvaraj Selvaraj} received the B.E. degree in electrical engineering from Anna University, India and the M.S. degree in systems, control, and mechatronics from Chalmers University of Technology, Sweden, in 2011 and 2016, respectively.  Between 2011 and 2014, he was with ABB, developing control software for industrial automation. Between 2016 and 2017, he was with Volvo Car Corporation. Since 2017, he is working towards the Ph.D. degree at Chalmers and Zenseact, Sweden. His research focuses on the application of formal methods for automated driving.  
\end{IEEEbiography}

\begin{IEEEbiography}[{\includegraphics[width=1in,height=1.2in,clip,keepaspectratio]{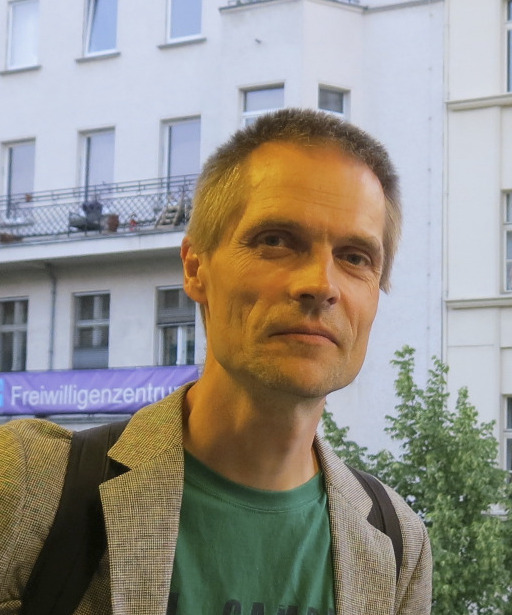}}]{Wolfgang Ahrendt} is professor in Computer Science at Chalmers University of Technology in Gothenburg, Sweden. He obtained a Ph.D. in Computer Science from the University of Karlsruhe in 2001. His main interests and contributions are in deductive verification of software, in runtime verification, and combinations of static verification with runtime verification and testing. Recent application areas of his work include safety of blockchain applications and automotive software safety.
\end{IEEEbiography}

\begin{IEEEbiography}[{\includegraphics[width=1in,height=1.2in,clip,keepaspectratio]{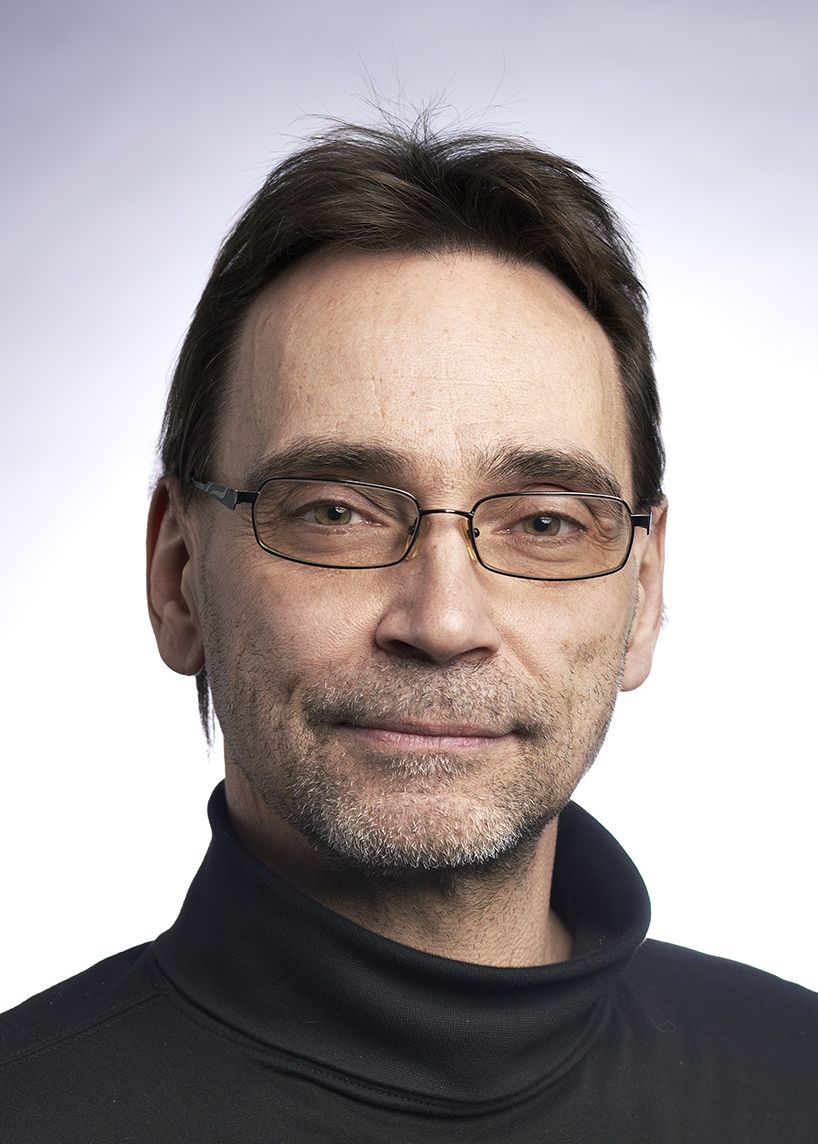}}]{Martin Fabian} is since 2014 full Professor in Automation and Head of the Automation Research group at the Department of Electrical Engineering, Chalmers University of Technology. He received his PhD in Control Engineering from Chalmers University of Technology in 1995. His research interests include formal methods for automation systems in a broad sense, merging the fields of Control Engineering and Computer Science. He has authored more than 200 publications, and is co-developer of the formal methods tool Supremica, which implements  state-of-the-art algorithms for supervisory control synthesis.
\end{IEEEbiography}
\vspace{11pt}

\vfill

\end{document}